\documentclass[journal]{IEEEtran}

\usepackage{algorithm}
\usepackage{algpseudocode}
\usepackage{mathdots}
\usepackage{amsmath,amsfonts,amssymb,amsthm}
\usepackage{mathrsfs}
\usepackage{comment,blkarray}
\usepackage{multirow,bigdelim}
\usepackage{cite}
\usepackage{tikz}
\usetikzlibrary{arrows,automata}
\usepackage[latin1]{inputenc}
\usepackage{verbatim}
\usepackage{graphicx}
\usepackage{cases}
\usepackage{booktabs}
\usepackage{caption}												
\usepackage{subcaption}	
\usepackage{etoolbox}

\usepackage{enumitem}
\usepackage{url} 
\makeatletter

\newcommand{\Rmnum}[1]{\expandafter\@slowromancap\romannumeral #1@}

\definecolor{revision}{rgb}{0,0,1} 

\newif\if@borderstar
\def\bordermatrix{\@ifnextchar*{%
  \@borderstartrue\@bordermatrix@i}{\@borderstarfalse\@bordermatrix@i*}%
}
\def\@bordermatrix@i*{\@ifnextchar[{\@bordermatrix@ii}{\@bordermatrix@ii[()]}}
\def\@bordermatrix@ii[#1]#2{%
\begingroup
  \m@th\@tempdima8.75\p@\setbox\z@\vbox{%
    \def\cr{\crcr\noalign{\kern 2\p@\global\let\cr\endline }}%
    \ialign {$##$\hfil\kern 2\p@\kern\@tempdima & \thinspace %
    \hfil $##$\hfil && \quad\hfil $##$\hfil\crcr\omit\strut %
    \hfil\crcr\noalign{\kern -\baselineskip}#2\crcr\omit %
    \strut\cr}}%
  \setbox\tw@\vbox{\unvcopy\z@\global\setbox\@ne\lastbox}%
  \setbox\tw@\hbox{\unhbox\@ne\unskip\global\setbox\@ne\lastbox}%
  \setbox\tw@\hbox{%
    $\kern\wd\@ne\kern -\@tempdima\left\@firstoftwo#1%
    \if@borderstar\kern2pt\else\kern -\wd\@ne\fi%
    \global\setbox\@ne\vbox{\box\@ne\if@borderstar\else\kern 2\p@\fi}%
    \vcenter{\if@borderstar\else\kern -\ht\@ne\fi%
    \unvbox\z@\kern -\if@borderstar2\fi\baselineskip}%
    \if@borderstar\kern -2\@tempdima\kern2\p@\else\,\fi\right\@secondoftwo#1 $%
  }\null \;\vbox{\kern\ht\@ne\box\tw@}%
\endgroup
}
\makeatother

\newtheorem{thm}{Theorem}
\newtheorem{lemma}{Lemma}
\newtheorem{cor}{Corollary}

\newtheorem{example}{Example}
\newtheorem{conjecture}{Conjecture}
\newtheorem{defn}{Definition}
\allowdisplaybreaks

\begin{document}

\title{A Construction of Pairwise 
  Co-prime Integer Matrices of Any Dimension and Their Least Common Right Multiple}

\author{Guangpu Guo and Xiang-Gen Xia, \IEEEmembership{Fellow}, \IEEEmembership{IEEE}

\thanks{G. Guo and X.-G. Xia are with the Department of Electrical and Computer Engineering,
  University of Delaware, Newark, DE 19716, USA
  (e-mails: guangpu@udel.edu and xxia@ee.udel.edu).
  This work was supported in part by the National Science Foundation (NSF) under Grant CCF-2246917.
}
}

\maketitle

\begin{abstract}
  Compared with co-prime integers, co-prime integer matrices are more challenging due to the non-commutativity. In this paper, we present a new family of pairwise co-prime integer matrices of any dimension and large size. These matrices are non-commutative and have low spread, i.e., their ratios of peak absolute values to mean absolute values (or the smallest non-zero absolute values) of their components are low.
  When  matrix dimension is larger than $2$, this family of matrices  differs from the existing families, such as circulant, Toeplitz matrices, or triangular matrices, and therefore, offers more varieties in applications. 
  In this paper, we first prove the pairwise coprimality of the constructed matrices, then determine their determinant absolute values, and their least common right multiple (lcrm) with a closed and simple form.
  We also analyze their sampling rates when these matrices are used as sampling matrices for a multi-dimensional signal. The proposed family
  of pairwise co-prime integer matrices may have applications in
  multi-dimensional Chinese remainder theorem (MD-CRT) that can be used to determine integer vectors from their integer vector remainders modulo a set of integer matrix moduli, and also in
  multi-dimensional sparse sensing and multirate systems.
\end{abstract}

\begin{IEEEkeywords}
  Pairwise co-prime integer matrices, least common right multiple (lcrm),
   Smith  form, 
   Chinese remainder theorem (CRT), multi-dimensional CRT (MD-CRT),
   multi-dimensional sampling. 
 
\end{IEEEkeywords}

\IEEEpeerreviewmaketitle


\section{Introduction}\label{s1}

It is well-known that a family of pairwise co-prime integers,
i.e., every pair of integers in the family are co-prime, have important applications in, such as, Chinese remainder theorem (CRT) \cite{crt, crt1} that has many applications in, for example, cryptography and coding theory \cite{crt1, crt4, crt_integer1}, and signal processing \cite{crt2,conv1,crt_integer1,radar_book,sg,gangli1,fft2,wenchaoli,radeee,xiaoli,congling1,congling2,congling3,lugan2}. 
Similarly, pairwise co-prime integer matrices have applications in multi-dimensional CRT (MD-CRT) \cite{MD1, MD2} that
can be used to determine  integer vectors from their  integer vector
remainders modulo a set of  integer matrix moduli.
Note that when the integer matrix moduli can be diagonalized simultaneously, MD-CRT  had appeared in earlier literature \cite{jiawenxian, PPV1}.
However, different from co-prime integers, due to the non-commutativity of matrices, co-prime integer matrices are much more challenging. 

Co-prime integer matrices have been studied
  in \cite{PPV2,pp_matrix1,pp_matrix2,pp_matrix3,coprime_pp,pp_nested}  with applications
in multi-dimensional sparse sensing and multi-dimensional multirate
systems. Most studies in \cite{PPV2,pp_matrix1,pp_matrix2,pp_matrix3,coprime_pp,pp_nested}
are for $2$ by $2$ circulant integer
matrices and their variants and commutative integer matrices, $3$ by $3$
circulant integer matrices, Toeplitz integer matrices, triangular integer matrices and their adjugate matrices.
In particular, a necessary and sufficient condition for two $2\times 2$ integer
matrices are co-prime was obtained in \cite{pp_matrix3},
which is easy to check. 

  Similar to the conventional CRT, in MD-CRT
  pairwise co-prime integer matrices as matrix moduli may play
  an important role
  as well to have a large range of uniquely determinable
  integer vectors from their integer vector remainders
  modulo the matrix moduli.
  In this paper, we present  a new family of pairwise co-prime integer
  matrices of any dimension and large size.
  They are non-commutative and  have  low
  spread, i.e., their ratios of peak absolute values to mean absolute
  values (or the smallest non-zero absolute values) of their components are low.
  We first prove the pairwise coprimality of the matrices in the family
  and then determine their determinant absolute values and
  also their least common right multiple (lcrm) with a closed and simple form.

  When matrix dimension is $2$, the family of co-prime integer
  matrices we construct
  in this paper happen to be a set of Toeplitz integer matrices, and 
  satisfy  the necessary and sufficient condition
  for two $2\times 2$ Toeplitz integer matrices to be co-prime
  obtained in \cite{pp_matrix3}. 
When  matrix dimension is larger than $2$,
the family of co-prime integer matrices we construct in this paper are
much different from those in \cite{PPV2,pp_matrix1,pp_matrix2,pp_matrix3,coprime_pp,pp_nested}. The
key differences are that our  construction of pairwise co-prime
integer matrices in this paper is: 
 i) for any dimension, ii) of large size, iii) not
pairwise commutative, iv) not circulant, and v) not Toeplitz
or triangular matrices. 
In addition, as mentioned earlier, we  determine
the lcrm of the family of pairwise co-prime matrices constructed
in this paper, including the family of $2\times 2$ co-prime integer matrices,
which has not been addressed in any existing literature.

  Note that
  the determinant absolute values of the matrix moduli
  are the sampling rates, i.e., the number of
  sampled points per unit spatial volume, 
  using these matrices as sampling matrices \cite{MDSP, jiawenxian, smith4, PPV1, MD1, MD2} for a multi-dimensional  signal.
  Also, an lcrm $\mathbf{R}$  of the matrix moduli 
  determines the range $\mathcal{N}(\mathbf{R})$
  detailed in 1) in Section \ref{s2}, called the fundamental parallelpiped (FPD)
  of $\mathbf{R}$ \cite{smith4},  of 
  the uniquely determinable integer vectors from their 
  integer vector remainders modulo the matrix moduli \cite{MD1}.
  The determinant absolute value $|\det(\mathbf{R})|$, called the {\em dynamic
    range}, 
  is the number of these uniquely determinable integer vectors, which
  is given  with the specified lcrm $\mathbf{R}$ in this paper 
   for our newly constructed family of
  integer matrices as matrix moduli. 
  
  In this paper, we also show that the sampling rates of our newly
  proposed pairwise co-prime integer matrices as (non-separable)
  sampling matrices in {\em each dimension} are much smaller than
  the maximal ones of the necessary sampling rates  of the conventional
  one dimensional samplings using diagonal (separable) integer
  sampling matrices, 
when their maximal determinant
  absolute values and the
  determinant absolute values of their lcrm matrices, i.e.,
  their dynamic ranges,
  are the same.  This is an advantage of non-separable
  sampling over separable sampling for a multi-dimensional signal.

  This paper is organized as follows.
  In Section \ref{s2}, we 
briefly introduce some necessary notations and preliminaries on
integer matrices including MD-CRT. 
In Section \ref{s3}, we present a novel family of
pairwise co-prime integer matrices of any dimension.
In Section \ref{s4}, we
prove the pairwise coprimality of the integer
matrices in the  constructed family.
In Section \ref{s5}, we determine the determinants of the integer
matrices in the constructed family and their lcrm,
and also analyze their sampling rates.
In Section \ref{s6}, 
we conclude this paper. 

\section{Some Necessary Notations and Preliminaries on Integer Matrices}\label{s2}

$\mathbb{Z}$ denotes the set of all integers and $\mathbb{R}$ denotes the
set of all real numbers. All vectors, such as $\mathbf{n}$, $\mathbf{f}$
and $\mathbf{r}$, 
and matrices, such as $\mathbf{M}$, $\mathbf{N}$ and  $\mathbf{P}$,
in this paper
are $D$ dimensional integer vectors and $D\times D$ dimensional
integer matrices, respectively, i.e.,
$\mathbf{n}$, $\mathbf{f}$, $\mathbf{r}\in  \mathbb{Z}^D$
and $\mathbf{M}$, $\mathbf{N}$, $\mathbf{P}\in \mathbb{Z}^{D\times D}$, unless
otherwise specified.
$\mathbf{I}$ is the $D\times D$ identity matrix, and $\mathbf{0}$
is the all $0$ matrix or vector. $\det(\mathbf{M})$
denotes the determinant of matrix $\mathbf{M}$, and $^{\top}$ stands for
the transpose. And diag stands for a $D\times D$ diagonal matrix. 
For a set $\mathcal{S}$, its cardinality is denoted by
$|\mathcal{S}|$. For two positive integers $n$ and $m$, the remainder of $n$ modulo $m$ is denoted by $\langle n \rangle_{m}$. Below we introduce some necessary concepts
on integer matrices and for details, see, for example,
\cite{matrix,MDSP,jiawenxian,smith2, smith3, smith4,PPV1,PPV2,pp_matrix1,pp_matrix2,pp_matrix3,coprime_pp,pp_nested,MD1, MD2}. These definitions, when reduced to the one-dimensional case, do not affect any of the classical results related to co-prime integers.  

\begin{itemize}
\item[1)] \textbf{Set} $\mathcal{N}(\mathbf{M})$
  called the fundamental parallelpiped (FPD)
  of $\mathbf{M}$ \cite{smith4}: Given a $D \times D$ nonsingular integer matrix $\mathbf{M}$, set $\mathcal{N}(\mathbf{M})$ is defined as
  the following set of integer vectors: 
\begin{equation}\label{1}
\mathcal{N}(\mathbf{M}) = \left\{ \mathbf{k} \mid \mathbf{k} = \mathbf{M} \mathbf{x}, \mathbf{x} \in [0,1)^D \text{ and } \mathbf{k} \in \mathbb{Z}^D \right\}. 
\end{equation}
The number of elements in $\mathcal{N}(\mathbf{M})$ is equal to
the absolute value of the determinant of matrix $\mathbf{M}$, i.e.,
$|\det(\mathbf{M})|$, \cite{MDSP, smith4}.
The FPD of $$\mathbf{N}=\begin{pmatrix}
        2 &3\\ 1&4
    \end{pmatrix}$$ is shown in Fig. 1, where the dashed edges and hollow vertices are not part of the FPD. We refer the reader to [24] for more details about FPD.
\begin{figure}[htbp]
        \centering
        \includegraphics[width=\columnwidth]{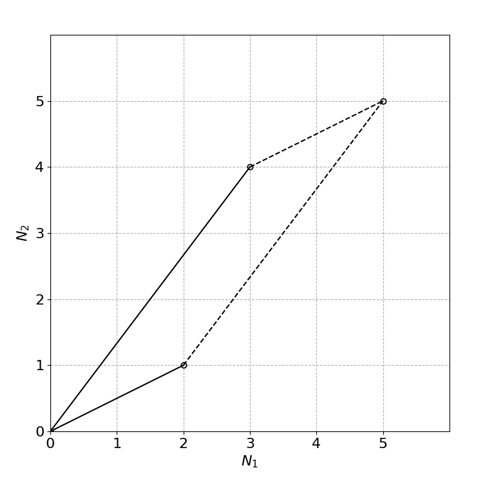}
        \caption{FPD of $\mathbf{N}$}
        \label{fig:fpd}
\end{figure}

\item[2)] \textbf{Unimodular matrix}: A square integer matrix is called unimodular if its determinant is $1$ or $-1$.

\item[3)] \textbf{Divisor and greatest common left divisor (gcld)}: A nonsingular integer matrix $\mathbf{A}$ is a left divisor of an integer matrix $\mathbf{M}$ if $\mathbf{A}^{-1}\mathbf{M}$ is an integer matrix. If $\mathbf{A}$ is a left divisor of each of all $L \geq 2$ integer matrices $\mathbf{M}_1, \mathbf{M}_2, \dots, \mathbf{M}_L$, it is called a common left divisor (cld) of $\mathbf{M}_1, \mathbf{M}_2, \dots, \mathbf{M}_L$. Moreover, if any other cld is a left divisor of $\mathbf{A}$, then $\mathbf{A}$ is a {\em greatest common left divisor} (gcld) of $\mathbf{M}_1, \mathbf{M}_2, \dots, \mathbf{M}_L$. 

\item[4)] \textbf{Co-prime matrices}: Two $D\times D$ integer matrices are left co-prime (or simply co-prime in this paper) if their gcld is a unimodular matrix. For two integer matrices $\mathbf{M}$ and $\mathbf{N}$,  they are left co-prime if and only if
  the Smith form \cite{smith1, smith2} of the combined $D\times 2D$ integer matrix
  $(\mathbf{M} \ \mathbf{N})$ is $(\mathbf{I} \ \mathbf{0})$. In addition, an equivalent necessary and sufficient condition on all the $D\times D$  minors
  of matrix  $(\mathbf{M} \ \mathbf{N})$ was proposed in \cite{pp_matrix3}. Also, it is not hard to see that if the determinant absolute values of two integer
  matrices are co-prime, these two integer
  matrices are co-prime \cite{pp_matrix1}.   
Note that in this paper, only left coprimality is considered. 

\item[5)] \textbf{Multiple and least common right multiple (lcrm)}: A nonsingular integer matrix $\mathbf{A}$ is a right multiple of an integer matrix $\mathbf{M}$, if there exists a nonsingular integer matrix $\mathbf{P}$ such that $\mathbf{A} = \mathbf{M}\mathbf{P}$. If $\mathbf{A}$ is a right multiple of each of all $L \geq 2$ integer matrices $\mathbf{M}_1, \mathbf{M}_2, \dots, \mathbf{M}_L$, $\mathbf{A}$ is called a common right multiple (crm) of $\mathbf{M}_1, \mathbf{M}_2, \dots, \mathbf{M}_L$. Additionally, $\mathbf{A}$ is a {\em least common right multiple} (lcrm) of $\mathbf{M}_1, \mathbf{M}_2, \dots, \mathbf{M}_L$, if any other crm of them,
  is a right multiple of $\mathbf{A}$. If $\mathbf{A}$ is an lcrm of $\mathbf{M}_1,\cdots,\mathbf{M}_L$, $\mathbf{AU}$ is also an lcrm of them when $\mathbf{U}$ is a unimodular matrix, which means that lcrm is not unique but the absolute determinant value of lcrm is unique. And this absolute determinant value is the minimum one of that of all the crms of these matrices. Although the lcrm of $\mathbf{M}_1, \mathbf{M}_2, \dots, \mathbf{M}_L$ is not unique, we use lcrm($\mathbf{M}_1, \mathbf{M}_2, \dots, \mathbf{M}_L$) to denote one fixed representative. This helps simplify later expressions.
  
  From this definition, it is not hard to see
  that for any groups of $D\times D$
  integer matrices $\mathbf{M}_{1,1},\cdots,\mathbf{M}_{1,L_1},
  \cdots, \mathbf{M}_{k,1},\cdots,\mathbf{M}_{k,L_k}$,
  we have
 \begin{equation}\label{lcrm1}
 \begin{aligned}
    &\text{lcrm}(\mathbf{M}_{1,1},\cdots,\mathbf{M}_{1,L_1},
    \cdots, \mathbf{M}_{k,1},\cdots,\mathbf{M}_{k,L_k}) \\
    = &\text{lcrm}( \text{lcrm}(\mathbf{M}_{1,1},\cdots,\mathbf{M}_{1,L_1}),
    \cdots,\\
    & \quad \quad \quad \quad \text{lcrm}(\mathbf{M}_{k,1},\cdots,\mathbf{M}_{k,L_k})).
 \end{aligned}
 \end{equation}
\item[6)] \textbf{Division representation for integer vectors}:
Given a nonsingular integer matrix \( \mathbf{M} \in \mathbb{Z}^{D \times D} \), any integer vector \( \mathbf{f} \in \mathbb{Z}^D \) can be uniquely decomposed as:
\[
\mathbf{f} = \mathbf{M} \mathbf{n} + \mathbf{r},
\]
where \( \mathbf{r} \in \mathcal{N}(\mathbf{M}) \) and \( \mathbf{n} \in \mathbb{Z}^D \). In modulo form, this is represented as:
\[
\mathbf{f} \equiv \mathbf{r} \mod \mathbf{M},
\]
where \( \mathbf{M} \) is a modulus, and \( \mathbf{n} \) and \( \mathbf{r} \) are the folding integer vector and the integer  vector remainder  of \( \mathbf{f} \) modulo \( \mathbf{M} \), respectively.

\item[7)] \textbf{Multi-dimensional undersampling}: Consider the following
  multi-dimensional harmonic signal, \cite{jiawenxian, smith4, PPV1, MD1, MD2},
\begin{equation}\label{3}
x(\mathbf{t}) = a\exp{(j2\pi \mathbf{f}^\top \mathbf{t})} + \omega(\mathbf{t}), \quad \mathbf{t}\in \mathbb{R}^{D},
\end{equation}
where $a$ is an unknown amplitude  and $\mathbf{f} \in \mathbb{Z}^{D}$ is an
unknown $D$ dimensional frequency of the signal, and $\omega(\mathbf{t})$
is an additive noise. We want to determine the $D$ dimensional integer
frequency vector $\mathbf{f}=[N_1,N_2,\cdots,N_D]^\top$
from (possibly multiple) undersampled $D$
dimensional signals of  $x(\mathbf{t})$ with  low sampling rates, 
where all $N_i$ are assumed positive integers and
some of them may be large, i.e., large frequencies.

We use $L$ many $D\times D$ nonsingular integer matrices
$\mathbf{M}_1,\mathbf{M}_2,\cdots,\mathbf{M}_L$ to sample the multi-dimensional
signal in (\ref{3}) in the following sense:
\begin{equation}\label{md-undersample}
x_i[\mathbf{n}] = a\exp{(j2\pi \mathbf{f}^\top \mathbf{M}_i^{-\top} \mathbf{n})} + \omega[\mathbf{M}_i^{-\top} \mathbf{n}],
\end{equation}
$\mathbf{n}\in \mathbb{Z}^{D}$, $1 \leq i \leq L$, where $\mathbf{M}_1,\mathbf{M}_2,\cdots,\mathbf{M}_L$ are 
called {\em sampling matrices}. For each sampling matrix $\mathbf{M}_i$,
 there are  $|\det (\mathbf{M}_i)|$ many sampled
points per unit spatial volume of $\mathbb{R}^{D}$ as we can see
from (\ref{1}), which is called the {\em sampling rate}
(or sampling density) of sampling matrix $\mathbf{M}_i$ for a multi-dimensional signal. 

Next, by performing the  multi-dimensional DFT (MD-DFT) to each $x_i[\mathbf{n}]$ with respect to $n \in \mathcal{N}(\mathbf{M}_i^{\top})$, we have, for $k\in \mathcal{N}(\mathbf{M}_i)$,
\begin{equation}\label{mddft1}
\begin{aligned}
    &X_i(\mathbf{k}) \\ 
    = &a \sum_{\mathbf{n} \in \mathcal{N}(\mathbf{M}_i^{\top})} \exp(j 2 \pi \mathbf{f}^{\top} \mathbf{M}_i^{-\top} \mathbf{n}) \exp(-j 2 \pi \mathbf{k}^{\top} \mathbf{M}_i^{-\top} \mathbf{n})\\
   &+ \Omega_i(\mathbf{k}), \quad 1 \leq i \leq L,
\end{aligned}
\end{equation}
which arrives at 
\begin{equation}\label{mddft}
    X_i(\mathbf{k}) = a |\det(\mathbf{M}_i)| \delta(\mathbf{k} - \mathbf{r}_i) + \Omega_i(\mathbf{k}), \quad 1 \leq i \leq L,
\end{equation}
where $\mathbf{r}_i$ is the integer vector remainder of
the integer frequency vector $\mathbf{f}$ modulo $\mathbf{M}_i$,
and $\delta (\mathbf{n})$ is the discrete delta function that is $1$ when
$\mathbf{n}={\bf 0}$ and $0$ otherwise.
Although the sampling matrices employed in our framework are not necessarily diagonal (i.e., separable sampling), we can obtain their corresponding diagonal matrices by calculating their Smith forms. By applying appropriate input and output signal index transformations, equation (\ref{mddft1}) can be reformulated into an MD-DFT based on diagonal matrices, i.e., the separable case. Then, we can use the Fast Fourier Transform (FFT) for computational acceleration on each dimension, for more details, see \cite{jiawenxian}.
From (\ref{mddft}), one can detect the integer vector remainders
$\mathbf{r}_i \equiv \mathbf{f} \mod \mathbf{M}_i$, $1\leq i\leq L$.
Now the question becomes how to determine the integer frequency vector
$\mathbf{f}$ from these detected integer vector remainders. This can be
solved by using MD-CRT \cite{MD1}.

\item[8)] \textbf{MD-CRT}: [MD-CRT for integer vectors \cite{MD1}]
  Given $L$ matrix moduli $\mathbf{M}_i$ for $1 \leq i \leq L$, which are arbitrary nonsingular integer matrices, let $\mathbf{R}$ be any lcrm of them. For an integer vector $\mathbf{n} \in \mathbb{Z}^D$,  
  it can be uniquely determined  from its $L$ integer vector remainders $\mathbf{r}_i \equiv \mathbf{n} \mod \mathbf{M}_i$, $1\leq i\leq L$, if
  $\mathbf{n} \in \mathcal{N}(\mathbf{R})$.

  A detailed determination algorithm  can be found in \cite{MD1}.
  From this MD-CRT, the range of the uniquely determinable
  integer  vectors $\mathbf{f}$ is $\mathcal{N}(\mathbf{R})$
  for an lcrm $\mathbf{R}$ of integer matrices
  $\mathbf{M}_i$, $1\leq i\leq L$, and the number of
  such uniquely determinable integer  vectors is $|\det(\mathbf{R})|$,
  which is called the {\em dynamic range} of the sampling
  matrices $\mathbf{M}_i$, $1\leq i\leq L$. Clearly one would like to
  have small sampling rates $|\det(\mathbf{M}_i)|$  and large 
  dynamic range $|\det(\mathbf{R})|$.
\item[9)] \textbf{Some applications of multi-dimensional undersampling and MD-CRT}: We present three cases where multi-dimensional undersampling applies. One is the conventional multi-dimensional sampling below the Nyquist rate. This has similar applications as in the one dimensional case, such as sensor networks, and also has applications in computational imaging \cite{compimage}, where the interested 2 dimensional frequencies are too high compared to the sampling rates of the sensors in a group of multiple scattered monitoring sensors with low functionalities, such as low sampling rates and low powers. \\ 
The second is in moving target parameter estimation in synthetic aperture radar (SAR) imaging, where a moving target speed and location parameters appear as frequency components in the radar return signals after some radar signal processing, such as range compression. In SAR imaging, the antenna arrays are fixed on the platform and the spatial sampling rate (corresponding to the time sampling rate) is determined by the fixed distance between adjacent antenna elements. When the target moves fast, the frequency components in the radar return signals may become too large compared to the fixed spatial sampling rate, which causes undersampling.
To address this, a method using two co-prime linear arrays was proposed in \cite{gangli1} to accurately estimate the parameters of fast-moving targets. However, the above linear antenna arrays are not spatially efficient for platforms with limited space, such as an aircraft. A natural way is to use planar antenna arrays, turning a 1D radar return signal into a 2D radar return signal, where co-prime linear arrays turn to co-prime integer matrices. In this setting, the results developed in this paper and the MD-CRT may play a key role. Notably, 2D co-prime planar arrays have already been applied in array signal processing, for example, for direction-of-arrival (DoA) estimation in \cite{sensor} where, although, the co-prime integer matrices used for co-prime arrays are diagonal (or separable).\\
The third application is in recent multi-channel self-reset anlog-to-digtal converter (SRADC)
for complex-valued bandlimited signals \cite{lugan2}. It is a special case of 2D-CRT with 2
dimensional co-prime integer matrix moduli \cite{xiaopingli}.

\end{itemize}

As we can see already from the above definitions, all matrix multiplications
in this paper are
in the sense of left side multiplications. Also, from the above coprimality of integer matrices, the elements of an integer matrix can be any integers
including negative integers, $1$, and $0$. Since a scalar integer can be
thought of as a special integer matrix, i.e., a $1\times 1$ integer matrix,
for the consistence with integer matrices, all integers are considered
for the coprimality. Integers $p$ and $q$ are co-prime if and only if
their gcd is $1$ or $-1$. This relaxation does not affect any results in this
paper.

\section{New Construction of Non-diagonal Pairwise Co-prime Integer Matrices
of Dimension $D$}\label{s3}

We first let 
 $1<q_1 < q_2 <\cdots <q_L$ be 
$L$ pairwise co-prime positive integers and 
 define the following set of $D\times D$ integer matrices 
\begin{equation}\label{14}
  \mathcal{S}_D = \{\mathbf{M} | \mathbf{M} \in \{0,\pm1,\pm 2,
  \cdots,\pm q_L\}^{D\times D} \}.
\end{equation}
We next present a method and an algorithm to construct a family of
pairwise co-prime integer matrices in the set $\mathcal{S}_D$. 
To do so, we first provide a definition.
Let $\mathcal{N}_D \stackrel{\Delta}{=} \{1, 2, \dots, D\}$. A permutation $\sigma$ of $\mathcal{N}_D$ is a one-to-one and onto mapping from $\mathcal{N}_D$ to itself. It can be represented by a vector $\sigma = (\sigma(1), \sigma(2), \dots, \sigma(D))$, where each $\sigma(i)\in \mathcal{N}_D$ and all $\sigma(i)$ are distinct. 

\begin{defn}
  For a given non-empty subset $\mathcal{K}$ of $ \mathcal{N}_D$, a feasible permutation set $\mathcal{P}_{f}(\mathcal{N}_D)$
  of $\mathcal{N}_D$
  is defined as a
  subset of all permutations of $\mathcal{N}_D$:
  \begin{equation}\label{301}
  \begin{aligned}
    \mathcal{P}_f(\mathcal{N}_D)=
    \{\sigma_j \mbox{ is a} \ &\mbox{permutation of }\mathcal{N}_D \mbox{ and its last
      element }\\ &\sigma_j(D)=j, 
    j \in \mathcal{K}\}.
    \end{aligned}
    \end{equation}
\end{defn}

From the above definition, it is not hard to see that
there are a total of 
$$\sum_{d=1}^{D} \binom{D}{d} ((D-1)!)^d$$
feasible permutation sets of $\mathcal{N}_D$.
Any non-empty subset of $d$ elements of $\mathcal{N}_D$
defines a feasible permutation set $\mathcal{P}_f(\mathcal{N}_D)$ of $\mathcal{N}_D$ and its size is $d$ as well. If the whole set $\mathcal{N}_D$ is taken in
defining a feasible permutation set, i.e., all the elements in $\mathcal{N}_D$
are taken as the last components of the permutations in  $\mathcal{P}_f(\mathcal{N}_D)$, the feasible permutation set has the largest size $D$.

For a given $\mathcal{P}_{f}(\mathcal{N}_D)$, we can construct a family of
$D\times D$ pairwise co-prime matrices as follows.

To construct a $D\times D$ integer matrix $\mathbf{M}$,  we begin by selecting a permutation from $\mathcal{P}_{f}(\mathcal{N}_D)$ to place entries of $1$ in
its specific positions. The construction details are  
\begin{itemize}
\item[1)] For any chosen permutation $$\sigma_j = (\sigma_j(1), \sigma_j(2), \cdots, \sigma_j(D))$$ from $\mathcal{P}_{f}(\mathcal{N}_D)$, we set the elements at the $(D-1)$ positions $(\sigma_j(1), \sigma_j(2))$, $(\sigma_j(2), \sigma_j(3))$, $\cdots$, $(\sigma_j(D-1), \sigma_j(D))$ of  matrix $\mathbf{M}$ to $1$.
\item[2)] Next, we choose the diagonal elements from the set of
  pairwise co-prime positive integers $\{q_1, q_2, \cdots, q_L\}$ with $q_1>1$.
  By choosing any integer $q_i$ from this set, we set all diagonal entries of  matrix $\mathbf{M}$  to $q_i$.
  \item[3)] Set all the  other elements of $\mathbf{M}$
    to $0$.
\end{itemize}
    This completes the construction of one $D\times D$ integer matrix. 
    For each permutation from $\mathcal{P}_{f}(\mathcal{N}_D)$, we can create $L$ distinct matrices by choosing different diagonal elements from the
    pairwise co-prime integer set  $\{q_1, q_2, \cdots, q_L\}$.
    Given a $\mathcal{P}_{f}(\mathcal{N}_D)$ with $|\mathcal{P}_{f}(\mathcal{N}_D)|=d$ for some positive integer $d$ with  $1 \leq d \leq D$, this approach allows us to construct a total of $dL$ many $D\times D$ integer matrices.
    In other words, for any given
    feasible permutation set $\mathcal{P}_{f}(\mathcal{N}_D)$ of size $d$,
    $1\leq d\leq D$, 
    we can construct a family of $dL$ many $D\times D$ integer matrices
    that will be shown pairwise co-prime later. 
The construction method is summarized in Algorithm \ref{alg:1}.

\begin{algorithm}
\caption{Pairwise Co-prime Matrices Construction Algorithm}
\label{alg:1}
\begin{algorithmic}
    \State \textbf{Input:} A set of pairwise co-prime positive integers $\{q_1, q_2, \cdots, q_L\}$ with $q_1>1$, dimension $D$ and a feasible permutation set $\mathcal{P}_{f}(\mathcal{N}_D)$ with $|\mathcal{P}_{f}(\mathcal{N}_D)|=d$
    \State \textbf{Output:} A set of $D\times D$ integer matrices $\{\mathbf{M}_1, \mathbf{M}_2, \cdots, \mathbf{M}_{Ld}\}$
    
    \State Initialize an empty list to store matrices
    
    \For{each permutation $\sigma_j = (\sigma_j(1), \sigma_j(2), \cdots, \sigma_j(D))$ from $\mathcal{P}_{f}(\mathcal{N}_D)$}
        \For{each $q_i \in \{q_1, q_2, \cdots, q_L\}$}
            \State Initialize a $D \times D$ matrix $\mathbf{M}$ with zeros
            
            \State Set the element at position $(\sigma_j(1), \sigma_j(1))$ of matrix $\mathbf{M}$ to $q_i$
            
            \For{$k = 2$ to $D$}
                \State Set the element at position $(\sigma_j(k-1), \sigma_j(k))$ of matrix $\mathbf{M}$ to $1$
                \State Set the element at position $(\sigma_j(k), \sigma_j(k))$ of matrix $\mathbf{M}$ to $q_i$
            \EndFor
            
            \State Add matrix $\mathbf{M}$ to the list of matrices
        \EndFor
    \EndFor
    
    \State \Return the list of matrices
\end{algorithmic}
\end{algorithm}

Let $\mathcal{P}_{f}(\mathcal{N}_D)$  be a feasible
permutation set of $\mathcal{N}_D$  of the largest size $D$, such as, 
 the set of all the cyclic permutations of $\mathcal{N}_D$:
\begin{equation}\label{cyc}
\begin{aligned}
    \mathcal{P}_{f}(\mathcal{N}_D) = \{(1,2,&\cdots,D-1,D),(2,3,\cdots,D,1),\cdots,\\
    &(D,1,\cdots,D-2,D-1)\}.
\end{aligned}
\end{equation}
Following the above construction method, for the sake of convenience,
let $\mathbf{M}_{i,j}$ denote the matrix that is constructed
by choosing $q_i$ on its diagonal and choosing the permutation $\sigma_j$ from $\mathcal{P}_{f}(\mathcal{N}_D)$ to position the $1$'s, for $1 \leq i \leq L$ and $1 \leq j \leq D$. We can see that in matrix $\mathbf{M}_{i,j}$, there are a total of $(D-1)$ many $1$'s, and at most a single $1$ per row and at most a single $1$ per column. Especially, the $j$-th row has no $1$ and the $\sigma_{j}(1)$-th column has no $1$. Furthermore, each row has at most
two non-zero elements $q_i$ and $1$ and each column has at most two
non-zero elements $q_i$ and $1$.

Each matrix $\mathbf{M}_{i,j}$ in the above construction can be represented as
\begin{equation}\label{repre}
   \mathbf{M}_{i,j} = q_i\mathbf{I} + \mathbf{A}_j, 
\end{equation}
where $\mathbf{A}_j$ is a binary matrix constructed by following the above
construction Steps 1) and 3) with the chosen permutation $\sigma_j$ from $\mathcal{P}_{f}(\mathcal{N}_D)$. From the representation in (\ref{repre}), it is not hard
to check that $\mathbf{M}_{i_1,j}$ and $\mathbf{M}_{i_2,j}$ are commutative, for any $1\leq i_1,  i_2 \leq L$ and $1 \leq j \leq D$. However,
matrices $\mathbf{M}_{i_1,j_1}$ and $\mathbf{M}_{i_2,j_2}$
for $1\leq j_1\neq j_2\leq D$ in the above construction are not commutative. 

 As an example of matrices $\mathbf{M}_{i,j}$, consider the case 
 when $D=4$ and the chosen feasible permutation set of $\mathcal{N}_4=\{1,2,3,4\}$ is
 $$
 \mathcal{P}_{f}(\mathcal{N}_4)=\{(4,2,3,1),(1,3,4,2),(4,1,2,3),(3,1,2,4)\}.
 $$
 Then, the matrices constructed by
the method (or Algorithm \ref{alg:1})  are:  for $ 1 \leq i \leq L$, 
\begin{equation}\notag
\begin{aligned} 
\mathbf{M}_{i,1} &=
    \begin{pmatrix}
        q_i & 0 & 0 & 0\\
        0 & q_i & 1 & 0\\
        1 & 0 & q_i & 0\\
        0 & 1 & 0 & q_i\\
    \end{pmatrix},
    \mathbf{M}_{i,2} =
    \begin{pmatrix}
        q_i & 0 & 1 & 0\\
        0 & q_i & 0 & 0\\
        0 & 0 & q_i & 1\\
        0 & 1 & 0 & q_i\\
    \end{pmatrix}, \\
    \mathbf{M}_{i,3} &=
    \begin{pmatrix}
        q_i & 1 & 0 & 0\\
        0 & q_i & 1 & 0\\
        0 & 0 & q_i & 0\\
        1 & 0 & 0 & q_i\\
    \end{pmatrix},
    \mathbf{M}_{i,4} =
    \begin{pmatrix}
        q_i & 1 & 0 & 0\\
        0 & q_i & 0 & 1\\
        1 & 0 & q_i & 0\\
        0 & 0 & 0 & q_i\\
    \end{pmatrix}.
\end{aligned}
\end{equation}
From this example, one can see that these matrices
are not circulant, Toeplitz, triangular, or their
variants as studied in \cite{pp_matrix3}.

When $D$ is an even number and the chosen feasible permutation set is
\begin{equation}\label{toep}
\begin{aligned}
    \mathcal{P}_{f}(\mathcal{N}_D) = \{(&1,2,\cdots,D),(2, {\langle2+2\rangle_{D+1}},\cdots,{\langle2D\rangle_{D+1}}),\\&\cdots,
    (D,{\langle D+D \rangle_{D+1}},\cdots,{\langle D^{2}\rangle_{D+1}})\},
\end{aligned}
\end{equation} 
the family constructed by the above method (or Algorithm \ref{alg:1}) happens to be a family of Toeplitz matrices. For example, when $D=4$ and the chosen feasible permutation set is
 \begin{align}\label{toep4}
   \mathcal{P}_{f}(\mathcal{N}_4) &=\{(1,2,3,4),(2,4,1,3),(3,1,4,2),(4,3,2,1)\} \notag \\ 
   &=\{(4,3,2,1),(3,1,4,2),(2,4,1,3),(1,2,3,4)\},
 \end{align}
the constructed matrices are: for $ 1 \leq i \leq L$, 
\begin{equation}\notag
\begin{aligned} 
\mathbf{M}_{i,1} &=
    \begin{pmatrix}
        q_i & 0 & 0 & 0\\
        1 & q_i & 0 & 0\\
        0 & 1 & q_i & 0\\
        0 & 0 & 1 & q_i\\
    \end{pmatrix},
    \mathbf{M}_{i,2} =
    \begin{pmatrix}
        q_i & 0 & 0 & 1\\
        0 & q_i & 0 & 0\\
        1 & 0 & q_i & 0\\
        0 & 1 & 0 & q_i\\
    \end{pmatrix}, \\
    \mathbf{M}_{i,3} &=
    \begin{pmatrix}
        q_i & 0 & 1 & 0\\
        0 & q_i & 0 & 1\\
        0 & 0 & q_i & 0\\
        1 & 0 & 0 & q_i\\
    \end{pmatrix},
    \mathbf{M}_{i,4} =
    \begin{pmatrix}
        q_i & 1 & 0 & 0\\
        0 & q_i & 1 & 0\\
        0 & 0 & q_i & 1\\
        0 & 0 & 0 & q_i\\
    \end{pmatrix},
\end{aligned}
\end{equation}
where the index $j$ in $\mathbf{M}_{i,j}$ corresponds to the last component in a permutation in  $\mathcal{P}_{f}(\mathcal{N}_4)$ in (\ref{toep4}).

When $D=2$, i.e., the two dimension case, the feasible
permutation set with the maximal size $D=2$ has only one possibility, i.e.,
the one in (\ref{toep}) or (\ref{cyc}), and therefore, 
all the constructed matrices
$\mathbf{M}_{i,j}$ happen to be Toeplitz as we will study more later for their
lcrm. In general, our constructed family is not a 
family of Toeplitz matrices and the above case with the 
special feasible permutation set is the only case of Toeplitz matrices. 
Furthermore, we do not use any property of Toeplitz matrices in the following
studies.

For the above constructed family $\{\mathbf{M}_{i,j}: \,\,1\leq i\leq L, 1\leq j\leq D\}$  of $D\times D$ integer matrices, we have the following results
about their pairwise coprimality, determinants,  lcrm, and  dynamic ranges.
As mentioned above for the $2$ dimensional case, 
 our  contructed $2\times 2$ integer matrices happen to be
Toeplitz, every pair of which indeed satisfy the necessary and
sufficient condition for them to be co-prime obtained in
\cite{pp_matrix3}.

\section{Pairwise Coprimality}\label{s4}

In this section, we show the pairwise coprimality of the integer
matrices in the family constructed in the previous section. To do so,
we first present a lemma. 

\begin{lemma}\label{lm:1}
  Let $m_1$ and $m_2$ be two non-zero integers with $gcd(m_1,m_2) = k$ for a positive
  integer $k$.
  For each $i$ with $1 \leq i \leq D$, we can obtain a new matrix $[k\mathbf{e}_i \quad \mathbf{0}]$ by performing elementary column transformations on
   matrix $[m_1\mathbf{e}_i \quad m_2\mathbf{e}_i]$, where $\mathbf{e}_i$ is
  the $D$-dimensional vector with the $i$-th component $1$ and
  the other components $0$.
\end{lemma}

\begin{proof}
  When $|m_1| = |m_2|$, by multiplying $1$ or $-1$ to each column of  matrix $[m_1\mathbf{e}_i \quad m_2\mathbf{e}_i]$, we can get a new matrix $[k\mathbf{e}_i \quad k\mathbf{e}_i]$ for $k=|m_1|=|m_2|$. Then, by multiplying $-1$ to the first column and adding it to the second column, we get  matrix $[k\mathbf{e}_i \quad \mathbf{0}]$.

  When $|m_1| \neq |m_2|$, without loss of generality, we assume that $|m_1| > |m_2|$. Apply the Euclidean algorithm to $m_1$ and $m_2$ and  assume that there are $L$ equations here to calculate the gcd of $m_1$ and $m_2$:
\begin{equation}\label{3.2}
    \begin{aligned}
    &m_1 = n_1m_2 + r_1, \quad 0 \leq r_1 < |m_2|, \\
    &m_2 = n_2r_1 + r_2, \quad 0 \leq r_2 < r_1, \\
    &r_1 = n_3r_2 + r_3, \quad 0 \leq r_3 < r_2, \\
    &\vdots \\
    &r_{L-3} = n_{L-1}r_{L-2} + k, \quad 0 \leq k< r_{L-2}, \\
    &r_{L-2} = n_{L}k.
\end{aligned}
\end{equation}

Consider the first equation in (\ref{3.2}), we can get a new matrix $[r_1\mathbf{e}_i \quad m_2\mathbf{e}_i]$ by multiplying the second column of
matrix $[m_1\mathbf{e}_i \quad m_2\mathbf{e}_i]$ by $-n_1$ and adding it to the first column. Then, consider the second equation in (\ref{3.2}), by multiplying the first column of  matrix $[r_1\mathbf{e}_i \quad m_2\mathbf{e}_i]$ by $-n_2$  and adding it to the second column, we can get a new matrix $[r_1\mathbf{e}_i \quad r_2\mathbf{e}_i]$. Following all $L$ equations in (\ref{3.2}), we eventually get a matrix with two columns $k\mathbf{e}_i$ and $\mathbf{0}$
in the form of either $[k\mathbf{e}_i \quad \mathbf{0}]$
or $[\mathbf{0} \quad  k\mathbf{e}_i]$, depending on  $L$ is odd or even.
So, by switching the order of  the two columns if necessary, we
get matrix $[k\mathbf{e}_i \quad \mathbf{0}]$.
\end{proof}

We now consider the pairwise coprimality of the constructed
family  of integer matrices $\{\mathbf{M}_{i,j}\}$. 

\begin{thm}\label{thm:3}
  The matrices $\mathbf{M}_{i,j}$, for $1 \leq i \leq L$ and $1 \leq j \leq D$, 
  are pairwise co-prime.
\end{thm}

\begin{proof}
    We divide this proof into two parts. First, we prove that $\mathbf{M}_{i_1,j}$ and $\mathbf{M}_{i_2,j}$ are co-prime, for any $1\leq i_1 \neq i_2 \leq L$ and $1 \leq j \leq D$. Second, we prove that $\mathbf{M}_{i_1,j_1}$ and $\mathbf{M}_{i_2,j_2}$ are co-prime, for any $1\leq i_1, i_2 \leq L$ and $1 \leq j_1 \neq j_2 \leq D$.

    We begin by the first part.     
    Consider two matrices $\mathbf{M}_{i_1,j}$ and $\mathbf{M}_{i_2,j}$ for any given $1\leq i_1 \neq i_2 \leq L$ and $1 \leq j \leq D$. Let $\sigma_j = (m_1,m_{2},\cdots,m_{D})$, where $m_D=j$, and $D\times 2D$ matrix 
    $\mathbf{A} = (\mathbf{M}_{i_1,j} \quad \mathbf{M}_{i_2,j})$. We then calculate the Smith  form of $\mathbf{A}$ step by step.

    In the first step, we can see that the $m_1$-th column of $\mathbf{A}$ is $q_{i_1}\mathbf{e}_{m_1}$, and the $(D+m_1)$-th column of $\mathbf{A}$ is $q_{i_2}\mathbf{e}_{m_1}$. As $q_{i_1}$ and $q_{i_2}$ are co-prime for $i_1\neq i_2$, we 
    can get the new $m_1$-th column $\mathbf{e}_{m_1}$
    and the new $(D+m_1)$-th column $\mathbf{0}$ by applying Lemma \ref{lm:1}.
    In the second step, we can use     vector $\mathbf{e}_{m_1}$ to
    eliminate $1$'s 
    in the positions $(m_1,m_2)$ and $(m_1,D+m_2)$ of matrix $\mathbf{A}$ and get the new $m_2$-th column $q_{i_1}\mathbf{e}_{m_2}$, and the new $(D+m_2)$-th column
    $q_{i_2}\mathbf{e}_{m_2}$. Similarly, we can get the new $m_2$-th column $\mathbf{e}_{m_2}$ and the new $(D+m_2)$-th column $\mathbf{0}$ by applying Lemma \ref{lm:1}. By continuing in this manner, we can always get $q_{i_1}\mathbf{e}_{m_n}$ and $q_{i_2}\mathbf{e}_{m_n}$ in the $m_n$-th and the $(D+m_n)$-th columns 
    in the $n$-th step, respectively, 
    for $3 \leq n \leq D$. Therefore, we can always get two new columns $\mathbf{e}_{m_n}$ and $\mathbf{0}$ in the $n$-th step, for $3 \leq n \leq D$. After $D$ steps, we get all vectors $\mathbf{e}_{n}$, for $1 \leq n \leq D$. Finally, by rearranging the newly obtained columns, we obtain the Smith  form of $\mathbf{A}$ is $(\mathbf{I} \quad \mathbf{0})$. This proves that $\mathbf{M}_{i_1,j}$ and $\mathbf{M}_{i_2,j}$ are co-prime, 
     for any $1\leq i_1 \neq i_2 \leq L$ and $1 \leq j \leq D$.
    
We next      prove the second part.
Consider two matrices $\mathbf{M}_{i_1,j_1}$ and $\mathbf{M}_{i_2,j_2}$ for any given $1\leq i_1, i_2 \leq L$ and $1 \leq j_1 \neq j_2 \leq D$. Let $\mathbf{B} = (\mathbf{M}_{i_1,j_1} \quad \mathbf{M}_{i_2,j_2})$.
We then calculate the Smith  form of $D\times 2D$ matrix 
$\mathbf{B}$.

    Let $\sigma_{j_2}=(m_1,m_{2},\cdots,m_{D})$, where $m_{D}=j_2$. We first show that $\mathbf{M}_{i_2,j_2}$ can be transformed to the form of
    \begin{equation}\label{form}
    \begin{pmatrix}
     0 & 0 & \cdots & 0 & 1 & q_{i_2} \\
    0 & 0 & \cdots & 1 & q_{i_2} & 0 \\
    0 & 0 & \iddots & q_{i_2} & 0 & 0 \\
    \vdots & \iddots & \iddots & \iddots & \vdots & \vdots \\
    1 & q_{i_2} & \iddots & 0 & 0 & 0 \\
    q_{i_2} & 0 & \cdots &0  & 0 & 0 \\   
    \end{pmatrix}
   \end{equation}
    by rearranging the rows and columns.

    We first rearrange the rows of $\mathbf{M}_{i_2,j_2}$.
    After performing some 
    proper row permutations,  we can  
    let the $m_{1}$-th row be the first row, the $m_{2}$-th row be the second row, the $m_{3}$-th row be the third row and so on until the $m_{D}$-th row be the last row. From the construction, the positions of $1$'s in $\mathbf{M}_{i_2,j_2}$ are $(m_1,m_2)$, $(m_2,m_3)$, $\cdots$, $(m_{D-1},m_D)$
    and the positions of $q_{i_2}$'s in $\mathbf{M}_{i_2,j_2}$ are $(m_1,m_1)$, $(m_2,m_2)$, $\cdots$, $(m_{D},m_D)$. In the newly obtained matrix, the positions of $1$'s  are $(1,m_2)$, $(2,m_3)$, $\cdots$, $(D-1,m_D)$ and
    the positions of $q_{i_2}$'s are $(1,m_1)$, $(2,m_2)$, $\cdots$, $(D,m_D)$. Then, we rearrange the columns of the newly 
    obtained matrix by implementing column permutations. 
    Let the $m_{D}$-th column be the first column, the $m_{D-1}$-th column be the second column and so on until the $m_{1}$-th column be the last column. Now, the positions of $1$'s are $(1,D-1)$, $(2,D-2)$, $\cdots$, $(D-1,1)$ and the positions of $q_{i_2}$'s are $(1,D)$, $(2,D-1)$, $\cdots$, $(D,1)$. It is
    the form  in (\ref{form}).

After performing 
the above elementary transformations on  
$\mathbf{B}$, we get a new matrix, denoted by $\mathbf{B}^{'}$, and the
second 
half of $\mathbf{B}^{'}$ has the form of (\ref{form}).
We claim that there must be a $1$ in the last row of $\mathbf{B}^{'}$.
Otherwise, the row in $\mathbf{M}_{i_1,j_1}$ and the row
in $\mathbf{M}_{i_2,j_2}$ that have
no $1$ have the same row number.
Since in this case, from the  construction Step 1),
we know that the $j_1$-th row of $\mathbf{M}_{i_1,j_1}$ has no $1$ and
the $j_2$-th row of $\mathbf{M}_{i_2,j_2}$ has no $1$,
since $\sigma_{j_1}(D)=j_1$ and $\sigma_{j_2}(D)=j_2$.
Because from the construction, there is only one row has no $1$ in a matrix,
we have $j_1 = j_2$, which leads to contradiction.
Also, from the construction, each row has one $q_{i_1}$.
Thus, there are a $1$ and a $q_{i_1}$ in the last row of $\mathbf{B}^{'}$.
Let the column containing the $1$ in the last row be the $k_1$-th column and 
 the column containing the $q_{i_1}$ in the last row be 
 the $k_2$-th column. Without loss of generality, we assume $k_1 < k_2$.
 If not, we can exchange these two columns.
    Therefore, the matrix $\mathbf{B}^{'}$ now has the form of (\ref{B'}) on the top of the next page.
\begin{figure*}[!t]
\begin{equation}\label{B'}
\begin{pmatrix}
     * & \cdots& * & \cdots & * & \cdots & * &  0 & 0 & \cdots & 0 & 1 & q_{i_2} \\
     * & \cdots& * & \cdots & * & \cdots & * & 0 & 0 & \cdots & 1 & q_{i_2} & 0\\
     * & \cdots& * & \cdots & * & \cdots & * & 0 & 0 & \iddots & q_{i_2} & 0 & 0 \\
     \vdots & \vdots & \vdots & \vdots & \vdots & \vdots & \vdots & \vdots & \iddots & \iddots & \iddots & \vdots & \vdots \\
     * & \cdots& * & \cdots & *& \cdots & * & 1 & q_{i_2} & \iddots & 0 & 0 & 0 \\
     * & \cdots&1  & \cdots &q_{i_1}  & \cdots & * & q_{i_2} & 0 & \cdots &0  & 0 & 0 \\
    \end{pmatrix}
\end{equation}
\end{figure*}

    We then calculate the Smith  form step by step. In the first step, we use the $(2D-1)$-th column to eliminate the other non-zero elements
    in the first row by performing elementary column transformations. Then, by using the newly obtained the first row to perform elementary row transformations, we can get the new $(2D-1)$-th column $\mathbf{e}_1$. By continuing in this manner, in each step $n$, for $2 \leq n \leq D-1$, we can always identify a $1$ in the $(2D-n)$-th position of the $n$-th row. Then, using the $(2D-n)$-th column, we perform elementary column transformations to eliminate the other non-zero elements in the $n$-th row. Then, by using the newly obtained the $n$-th row to perform elementary row transformations, we can eliminate the other non-zero elements in this column and get the new $(2D-n)$-th column $\mathbf{e}_{n}$. By now, we have get all vectors $e_{n}$, for $1 \leq n \leq D-1$.
    
    We now consider  the $(2D)$-th column, the $k_1$-th column,
    and the $k_2$-th column after the above $(D-1)$ steps. Based on the previous discussion, each time $1$ is used to eliminate the other non-zero elements in the same row, there exists an adjacent element $q_{i_2}$ in the next row of the same column with the chosen $1$. As a result, when we eliminate a non-zero element from the current row, a new non-zero element is introduced in the next row within the same column with this non-zero element.
    Therefore, in each step $n$, for $1 \leq n \leq D-1$, we need to eliminate the non-zero element $(-1)^{n-1}q_{i_2}^{n}$ in the $n$-th position of the $(2D)$-th column, so the $(2D)$-th column becomes $[0,\cdots,0,(-1)^{D-1}q_{i_2}^{D}]^{\top}$ after $(D-1)$ steps.

    From the construction, the $k_1$-th column of $\mathbf{B}^{'}$ must has a single $1$, a single $q_{i_1}$, and $(D-2)$ zeros. Assume this  $q_{i_1}$
    in the $k_1$-th column is in the $r_1$-th row, $1 \leq r_1 \leq D-1$.
    All the elements in the column are zero for the first $(r_1-1)$ rows, meaning that the first $(r_1-1)$ steps have no effect on this column.
    Similar to the $(2D)$-th column, from step $r_1$ through step $D-1$,
    in each  step we  apply elementary column transformations to eliminate the non-zero element that appears in this column.  
    Then, the new $k_1$-th column after the $(D-1)$ steps
    is
    $$[0,\cdots,0,1+(-1)^{D-r_1}q_{i_1}q_{i_2}^{D-r_1}]^{\top}.$$

    The $k_2$-th column of $\mathbf{B}^{'}$ may have or may not have a $1$. Next, we will discuss these two cases separately and prove that, in both cases, we can obtain $\mathbf{e}_D$ by performing elementary column transformations on these three columns.
    
    \textbf{Case 1}: the $k_2$-th column of $\mathbf{B}^{'}$ has no $1$. 

    In this case, the $k_2$-th column of $\mathbf{B}^{'}$ is $[0,\cdots,0,q_{i_1}]^{\top}$. Since there are no non-zero elements in the first $(D-1)$ positions, the $k_2$-th column remains $[0,\cdots,0,q_{i_1}]^{\top}$ after the $(D-1)$ steps. Then, we can get a column of $\mathbf{e}_D$ by multiplying the $k_2$-th column by $(-1)^{D-r_1+1}q_{i_2}^{D-r_1}$ and adding it to the $k_1$-th column.

    \textbf{Case 2}: the $k_2$-th column of $\mathbf{B}^{'}$ has a $1$.

    In this case, assume that the $1$ is in the $r_2$-th position of the $k_2$-th column, $1 \leq r_2 \leq D-1$. Note that $r_1\neq r_2$, otherwise, there must be $a_1$ and $a_2$ from $\{1,2,\cdots,D\}$ such that the elements at $(a_1,a_2)$ position and $(a_2,a_1)$ position of $\mathbf{M}_{i_1,j_1}$ are all $1$, which is impossible from the  construction. Similar to the $k_1$-th column, the newly obtained $k_2$-th column after the $(D-1)$ steps is $$[0,\cdots,0,q_{i_1}+(-1)^{D-r_2}q_{i_2}^{D-r_2}]^{\top}.$$ 

    Since the power of $q_{i_2}$ in the $(2D)$-th column is larger than that in the $k_2$-th column, we can obtain the new $(2D)$-th column $$[0,\cdots,0,(-1)^{r_2}q_{i_1}q_{i_2}^{r_2}]^{\top}$$ by multiplying $(-1)^{r_2}q_{i_2}^{r_2}$ to the $k_2$-th column and adding it to the $(2D)$-th column.

    If $r_2 \leq D-r_1$, we can get a column of  $\mathbf{e}_{D}$ by multiplying the $(2D)$-th column by 
    \[
    (-1)^{D-r_1-r_2+1}q_{i_2}^{D-r_1-r_2}
    \]
    and adding it to the $k_1$-th column.
    
    If $r_2 > D-r_1$, we multiply the $k_1$-th column by $$(-1)^{r_2-(D-r_1)+1}q_{i_2}^{r_2-(D-r_1)}$$ and add it to the $(2D)$-th column. We then get the new $(2D)$-th column $$[0,\cdots,0,(-1)^{r_2-(D-r_1)+1}q_{i_2}^{r_2-(D-r_1)}]^{\top}.$$

    If $r_2-(D-r_1) \leq D-r_1$, we can obtain a column of
    $\mathbf{e}_D$ by multiplying $$(-1)^{2(D-r_1)-r_2}q_{i_1}q_{i_2}^{2(D-r_1)-r_2}$$ to the $(2D)$-th column and adding it to the $k_1$-th column.

    If $r_2-(D-r_1) > D-r_1$, we then compare $r_2-(D-r_1)$ and $D-r_2$. If $r_2-(D-r_1) \leq D-r_2$, we can get $q_{i_1}\mathbf{e}_{D}$ by multiplying $$(-1)^{(D-r_2)-r_2+(D-r_1)}q_{i_2}^{(D-r_2)-r_2+(D-r_1)}$$ to the $(2D)$-th column and adding it to the $k_2$-th column. Then, we can obtain a column of $\mathbf{e}_{D}$ by multiplying $(-1)^{D-r_1+1}q_{i_2}^{D-r_1}$ to the $k_2$-th column and adding it to the $k_1$-th column.

    From the steps above, it can be observed that when the power of $q_{i_2}$ in the $(2D)$-th column exceeds that in the $k_1$-th column or $k_2$-th column, we can use elementary column transformations with the $k_1$-th column or
    the $k_2$-th column to reduce the power of $q_{i_2}$ in the $(2D)$-th column  by amount of $D-r_1$ or $D-r_2$, respectively. Since $D$ is a finite integer, after a finite number of steps, we are guaranteed to obtain a new $(2D)$-th column $$[0,\cdots,0,(-1)^{b_1}q_{i_1}q_{i_2}^{b_2}]^{\top}$$ with $b_2 \leq D-r_1$ or a new $(2D)$-th column $$[0,\cdots,0,(-1)^{b_3}q_{i_2}^{b_4}]^{\top}$$ with $b_4 \leq D-r_1$ or $b_4 \leq D-r_2$. Then, we can use this $(2D)$-th column to obtain
    a column of $\mathbf{e}_D$ by performing elementary column transformations.

    Thus,    combining the above two cases, we can always get a column of
    $\mathbf{e}_{D}$. Then, by eliminating the other elements in the $D$-th row and rearranging all the columns,  the Smith  form of $\mathbf{B}$ is $(\mathbf{I}\quad \mathbf{0})$. This means that $\mathbf{M}_{i_1,j_1}$ and $\mathbf{M}_{i_2,j_2}$ are co-prime, for any $1\leq i_1, i_2 \leq L$ and $1 \leq j_1 \neq j_2 \leq D$. This proves the second part. 

    By combining the above two parts, we have completed the proof. 
\end{proof}

We next show that the above constructed matrices $\mathbf{M}_{i,j}$
are still pairwise co-prime  when we change some elements of the matrices, which provides much more selections of pairwise co-prime integer
matrices.

\begin{cor}\label{cor:1}
  If the signs of any  (one or more) elements in any matrix within the family constructed by Algorithm \ref{alg:1} are changed, the modified integer
  matrices are still pairwise co-prime. 
\end{cor}

\begin{proof}
     For the proof  we only need to simply revise the proof of Theorem \ref{thm:3} and show that these changes don't affect their coprimality by proving that these changes don't affect the calculations of the Smith  forms in the proof of Theorem \ref{thm:3}.

     First, we show that if we change the signs of  any  elements in the matrix $\mathbf{A}$ of the proof of Theorem \ref{thm:3}, we can also obtain the Smith
     form of the newly obtained matrix, denoted by $\mathbf{A}_1$, as  $(\mathbf{I} \quad \mathbf{0})$.

     Compute the Smith  form of $\mathbf{A}_1$ by following the same steps used in  computing the Smith  form of $\mathbf{A}$. In each step $n$, for $1 \leq n \leq D$, we can always get $q_{i_1}\mathbf{e}_{m_n}$ and $q_{i_2}\mathbf{e}_{m_n}$ in the $m_n$-th and the $(D+m_n)$-th columns by multiplying $1$ or $-1$ to the $m_n$-th and the $(D+m_n)$-th column. Therefore, we can always get $\mathbf{e}_{m_n}$ in the $n$-th step. After $D$ steps, we get all vectors $\mathbf{e}_{n}$, for $1 \leq n \leq D$. Finally, by rearranging the newly obtained columns, we obtain the Smith  form of $\mathbf{A}_{1}$ as $(\mathbf{I} \quad \mathbf{0})$. This proves the first part.

     Second, we show that if we change the signs of any elements in the matrix $\mathbf{B}$ of the proof of Theorem \ref{thm:3}, we can also obtain the Smith  form of the newly obtained matrix, denoted by $\mathbf{B}_1$, as  $(\mathbf{I} \quad \mathbf{0})$.

     When we apply the same transformations to matrix $\mathbf{B}_1$ as we did to matrix $\mathbf{B}$ in the proof of Theorem \ref{thm:3}, we obtain a new matrix $\mathbf{B}_1^{'}$, which can also be obtained by changing the signs of any
     elements in matrix $\mathbf{B}^{'}$. Next, we calculate the Smith form of $\mathbf{B}_1^{'}$ with the same steps in Theorem \ref{thm:3}.

     In each step $n$, for $1\leq n \leq D-1$, we can always identify a $1$ in the $(2D-n)$-th position of the $n$-th row by multiplying $1$ or $-1$ to the $(2D-n)$-th column. Therefore, we can always get $\mathbf{e}_{n}$ in the $n$-th step.
     
     After $(D-1)$ steps, the $(2D)$-th column becomes
$$
     [0,\cdots,0,q_{i_2}^{D}]^{\top} \mbox{ or }
     [0,\cdots,0,-q_{i_2}^{D}]^{\top},
     $$
     the $k_1$-th column becomes
$$
     [0,\cdots,0,1+q_{i_1}q_{i_2}^{D-r_1}]^{\top} \mbox{ or }
     [0,\cdots,0,1-q_{i_1}q_{i_2}^{D-r_1}]^{\top}
     $$
     by multiplying it by $1$ or $-1$.
     If the $k_2$-th column of $\mathbf{B}_{1}^{'}$ has no $1$, it is
$$
     [0,\cdots,0,q_{i_1}]^{\top} \mbox{ or } [0,\cdots,0,-q_{i_1}]^{\top}
     $$
     after $D-1$ steps. We can get $\mathbf{e}_D$ by multiplying the $k_2$-th column by $q_{i_2}^{D-r_1}$ or $-q_{i_2}^{D-r_1}$ and adding it to the $k_1$-th column.
     If the $k_2$-th column of $\mathbf{B}_{1}^{'}$ has a $1$, the newly obtained $k_2$-th column after $(D-1)$ steps is
$$
     [0,\cdots,0,q_{i_1}+q_{i_2}^{D-r_2}]^{\top} \mbox{ or }
     [0,\cdots,0,q_{i_1}-q_{i_2}^{D-r_2}]^{\top}
$$
by multiplying it by $1$ or $-1$.
     Then, we can use the same steps as in the proof of
     Theorem \ref{thm:3} to reduce the power of $q_{i_2}$ in the $(2D)$-th column until we obtain a new $(2D)$-th column
     $$
     [0,\cdots,0,(-1)^{b_1}q_{i_1}q_{i_2}^{b_2}]^{\top}
     $$
     with $b_2 \leq D-r_1$ or a new $(2D)$-th column
     $$
     [0,\cdots,0,(-1)^{b_3}q_{i_2}^{b_4}]^{\top}
     $$
     with $b_4 \leq D-r_1$ or $b_4 \leq D-r_2$. Then, we can use this $(2D)$-th column to obtained $\mathbf{e}_D$ by performing elementary column transformations.
     Finally, by eliminating the other elements in the $D$-th row and rearranging all the columns, we get the Smith  form of $\mathbf{B}_{1}$ as $(\mathbf{I}\quad \mathbf{0})$. This proves the second part.
     
     Combining these two parts, the corollary is proved.
\end{proof}

\section{Determinants and Least Common Right Multiples}\label{s5}

In this section, we first determine the  determinants of the
integer matrices constructed in Section \ref{s3},
which correspond to the sampling rates as mentioned in Introduction 
using the sampling matrices $\mathbf{M}_{i,j}$
in the multi-dimensional sampling problem described in 7) and 8) in
Section \ref{s2}.

\begin{thm}\label{thm:2}
  The determinant of the matrix $\mathbf{M}_{i,j}$
is the product of all the diagonal elements, i.e.,
$\det(\mathbf{M}_{i,j})=q_i^D$,
for $1 \leq i \leq L$ and $1 \leq j \leq D$.
\end{thm}

\begin{proof}
  From the  construction Steps 1) and 2),
  the diagonals of $\mathbf{M}_{i,j}$ are all
  $q_i$ and the chosen permutation is $\sigma_j$. For convenience, let 
   $\sigma_j=(m_1,\cdots,m_{D-1},j)$. Next, we calculate the determinant of $\mathbf{M}_{i,j}$ by applying Laplace expansion.

  From the construction Step 1), we know that
  the $m_1$-th column has no $1$, which means that the $m_1$-th column of $\mathbf{M}_{i,j}$ is $q_i\mathbf{e}_{m_1}$. So, we expand the $\det{(\mathbf{M}_{i,j})}$ by the $m_1$-th column, and we have 
  \begin{align}
    \det(\mathbf{M}_{i,j}) &= (-1)^{m_1 + m_1} q_i \det({\mathbf{M}_{i,j}^{\backslash (m_1)}}) \notag \\
    &= q_i \det({\mathbf{M}_{i,j}^{\backslash (m_1)}}).\notag
  \end{align}
    where ${\mathbf{M}_{i,j}^{\backslash (m_1)}}$ is the $(D-1)\times(D-1)$ submatrix obtained by removing the $m_1$-th column and the $m_1$-th row of
    $\mathbf{M}_{i,j}$.

    After removing the $m_1$-th row  and the $m_1$-th column
    of $\mathbf{M}_{i,j}$,
    from the construction Step 1) we have that the elements in the $m_2$-th column of $\mathbf{M}_{i,j}$ are all zeros expect a single $q_i$,
    since the element $1$ located at the position of $(m_1,m_2)$ in $\mathbf{M}_{i,j}$ has been removed.
    Then, we can identify this column in  matrix ${\mathbf{M}_{i,j}^{\backslash (m_1)}}$. Without loss of generality, we assume that this column is the $k_1$-th column of ${\mathbf{M}_{i,j}^{\backslash (m_1)}}$, for $1\leq k_1 \leq D-1$.
    Besides,  all $q_i$'s of ${\mathbf{M}_{i,j}^{\backslash (m_1)}}$ are also in all the diagonals of ${\mathbf{M}_{i,j}^{\backslash (m_1)}}$, because  ${\mathbf{M}_{i,j}^{\backslash (m_1)}}$ is obtained by
    deleting the $m_1$-th row and the $m_1$-th column of $\mathbf{M}_{i,j}$. 
    We then expand $\det{({\mathbf{M}_{i,j}^{\backslash (m_1)}})}$ by the $k_1$-th column and
    have 
    \begin{align}
      \det{({\mathbf{M}_{i,j}^{\backslash (m_1)}})} &= (-1)^{k_1+k_1} q_i \det{({\mathbf{M}_{i,j}^{\backslash (m_1,m_2)}})} \notag\\ 
      &= q_i\det{({\mathbf{M}_{i,j}^{\backslash (m_1,m_2)}})}, \notag 
    \end{align}
    where ${\mathbf{M}_{i,j}^{\backslash (m_1,m_2)}}$ is the $(D-2)\times(D-2)$ submatrix obtained by removing the $k_1$-th column and the $k_1$-th row of ${\mathbf{M}_{i,j}^{\backslash (m_1)}}$,
    which is the same as removing the $m_1$-th and the $m_2$-th rows and the $m_1$-th and the $m_2$-th columns  of $\mathbf{M}_{i,j}$.

    Similarly, for each $3\leq n \leq D-2$, we can always identify a column of all zeros expect a single $q_{i}$ in the $k_{n-1}$-th column of matrix ${\mathbf{M}_{i,j}^{\backslash (m_1,\cdots,m_{n-1})}}$, which is obtained by removing all $m_l$-th rows and all $m_l$-th columns of $\mathbf{M}_{i,j}$ for  $1\leq l \leq n-1$.
    Actually, this column is the $m_n$-th column of the original matrix $\mathbf{M}_{i,j}$. Since all $q_i$'s are in the diagonals of
    matrix $\mathbf{M}_{i,j}$ and the deleted
    rows and columns have the same indices in all the steps
    to get the next new matrices, the remaining $q_i$'s are always
    in the diagonals of matrix ${\mathbf{M}_{i,j}^{\backslash (m_1,\cdots,m_{n-1})}}$. So, we can expand the $\det{({\mathbf{M}_{i,j}^{\backslash (m_1,\cdots,m_{n-1})}})}$ by this column and have 
    \begin{align}
        &\det{({\mathbf{M}_{i,j}^{\backslash (m_1,\cdots,m_{n-1})}})}\notag\\ =& (-1)^{k_{n-1}+k_{n-1}} q_i \det{({\mathbf{M}_{i,j}^{\backslash (m_1,\cdots,m_{n})}})} \notag\\
        =& q_i\det{({\mathbf{M}_{i,j}^{\backslash (m_1,\cdots,m_{n})}})}\notag,
    \end{align}
    where ${\mathbf{M}_{i,j}^{\backslash (m_1,\cdots,m_{n})}}$ is the $(D-n)\times(D-n)$ submatrix obtained by removing all the $m_l$-th columns and all the $m_l$-th rows of $\mathbf{M}_{i,j}$, for $1\leq l \leq n$. 

    Eventually, we can get that 
    \[
    \det{(\mathbf{M}_{i,j})} = q_i^{D-2} \det\left({\begin{bmatrix}
        q_i & 0\\ 1 &q_i\\ 
    \end{bmatrix}}\right)
    \]
    or
    \[
    \det{(\mathbf{M}_{i,j})} = q_i^{D-2} \det\left({\begin{bmatrix}
        q_i & 1\\ 0 &q_i\\ 
    \end{bmatrix}}\right),
    \]
    which means that $\det{(\mathbf{M}_{i,j})} = q_{i}^{D}$, i.e.,
    the product of all the diagonal elements of $\mathbf{M}_{i,j}$.
      This completes the proof of  Theorem \ref{thm:2}.
\end{proof}

Similar to Corollary \ref{cor:1}, from the above proof we immediately have
the following corollary.
\begin{cor}\label{cor:2}
  If the signs of any  (one or more) elements in any matrix $\mathbf{M}_{i,j}$
  within the family constructed by Algorithm \ref{alg:1} are changed,
the determinant of the modified matrix 
is the product of all the diagonal elements, i.e.,
$det(\mathbf{M}_{i,j})=\pm q_i^D$,
where $1 \leq i \leq L$ and $1 \leq j \leq D$.
\end{cor}

Since $q_1,q_2,\cdots,q_L$ are pairwise co-prime, we can directly obtain the following corollary using one result in 4) of Section \ref{s2}, i.e.,
two integer matrices are co-prime if  their determinants
are co-prime \cite{pp_matrix1}. 

\begin{cor}\label{cor3}
$\mathbf{M}_{i_1,j_1}$ and $\mathbf{M}_{i_2,j_2}$, for any $1 \leq i_1 \neq i_2 \leq L$ and $1 \leq j_1,j_2 \leq D$, are co-prime.
\end{cor}

This corollary directly leads to the proof of the first part
in the proof of Theorem \ref{thm:3} in the previous section, i.e.,
for each fixed $j$, $1\leq j\leq D$, the group of
matrices $\{\mathbf{M}_{i,j}, 1\leq i\leq L\}$ are pairwise co-prime,
which, however, has only $L$ integer matrices. One of our main
contributions of the constructed family $\{\mathbf{M}_{i,j}\}$ in
the previous section 
is being able  to
add $D-1$ many more integer matrices
with the same determinant 
to the pairwise co-prime integer matrix family  for each fixed $i$,
i.e., for each fixed integer $q_i$ in the known set of
co-prime integers $q_i$, $1\leq i\leq L$.
This leads to that the family of pairwise co-prime integer matrices
in our new construction  has $DL$ many matrices. 
In addition, it also implies that the coprimality of determinants
of two integer matrices is only a sufficient  but not necessary condition 
for  two integer matrices to be co-prime. 

With Theorem \ref{thm:2}, we immediately have
the following absolute value of the determinant
of the product of all the matrices in the constructed family:
\begin{equation}\label{determ1}
  \left| \det\left(\prod_{1\leq i\leq L, 1\leq j\leq D} \mathbf{M}_{i,j}\right) \right|
  =(q_1q_2\cdots q_L)^{D^2}.
  \end{equation}

It is known that in one dimensional case, for the given pairwise
co-prime integers $1<q_1<q_2<\cdots<q_L$, their least common multiple
(lcm) is their product $q_1q_2\cdots q_L$
that corresponds to the dynamic range of the conventional CRT
using $q_1,q_2,\cdots,q_L$ as moduli. This means
that all the nonnegative integers within the dynamic range can be
uniquely determined by using CRT from their remainders modulo moduli
$q_i$, $1\leq i\leq L$.
However,
due to the non-communtativity
of the constructed integer matrices $\mathbf{M}_{i,j}$, their
lcrm $\mathbf{R}$ may not be their product. Thus, it is not clear
whether the value in (\ref{determ1}) is the determinant
absolute value of their lcrm, which is the 
 the  dynamic range, i.e., the number of integer vectors
that can be uniquely determined by using MD-CRT and their
integer vector remainders  modulo the matrix moduli as mentioned
in 8) in Section \ref{s2}) 
similar to the conventional CRT \cite{MD1, MD2}. 
Next, we first determine an lcrm of the constructed family $\mathbf{M}_{i,j}$
for the two dimensional case, i.e., $D=2$,
without the need of symbolic computations.

The $2\times 2$ pairwise co-prime integer matrices
constructed by our proposed  method (or Algorithm \ref{alg:1})  are: for $ 1 \leq i \leq L$, 
\begin{equation}\label{D=2} 
\mathbf{M}_{i,1} =
    \begin{pmatrix}
        q_i & 0 \\
        1 & q_i \\
    \end{pmatrix},
    \mathbf{M}_{i,2} =
    \begin{pmatrix}
        q_i & 1 \\
        0 & q_i \\
    \end{pmatrix}.
\end{equation}
As mentioned before, 
the above $2\times 2$ integer matrices happen to be Toeplitz and
do satisfy the necessary and sufficient condition
for two $2\times 2$ integer (Toeplitz)  matrices to be co-prime obtained
in \cite{pp_matrix3}. 

We now calculate  $\text{lcrm}(\mathbf{M}_{i,1},\mathbf{M}_{i,2})$ following \cite{smith2,MD2}, for $ 1 \leq i \leq L$. First, we calculate
\[
\mathbf{M}_{i,1}^{-1}\mathbf{M}_{i,2} = 
\begin{pmatrix}
    1         &1/q_{i}\\
    -1/q_{i}  &(q_{i}^{2}-1)/q_{i}^{2}\\
\end{pmatrix}.
\]
The lcm of the denominators of all the elements in $\mathbf{M}_{i,1}^{-1}\mathbf{M}_{i,2}$ is $q_{i}^{2}$. One can easily check that the Smith form of
matrix $q_{i}^{2}\mathbf{M}_{i,1}^{-1}\mathbf{M}_{i,2}$ is 
\begin{equation*}
   \begin{aligned}
 &\mathbf{U}(q_{i}^{2}\mathbf{M}_{i,1}^{-1}\mathbf{M}_{i,2})\mathbf{V} \\ =&
\begin{pmatrix}
    0 & 1\\
    1 & q_{i}^{3}+q_{i}\\
\end{pmatrix}
\begin{pmatrix}
    q_{i}^{2}         &q_{i}\\
    -q_{i}  &q_{i}^{2}-1\\
\end{pmatrix}
\begin{pmatrix}
    -q_{i} & q_{i}^{2}-1\\
    -1 & q_{i}\\
\end{pmatrix} \\ =& \begin{pmatrix}
    1 &0\\
    0 &q_{i}^{4}\\
\end{pmatrix}.   
\end{aligned} 
\end{equation*}
Let
\[
\mathbf{\Lambda} = \frac{1}{q_i^2}\begin{pmatrix}
   1 &0\\
    0 &q_{i}^{4}\\ 
\end{pmatrix} = \begin{pmatrix}
   \frac{1}{q_i^2} & 0\\
   0 & {q_i^2}\\
\end{pmatrix},
\]
$\mathbf{\Lambda}_{\alpha}$ be a diagonal matrix whose diagonal elements are formed by taking the numerators of all the diagonal elements of $\mathbf{\Lambda}$ and $\mathbf{\Lambda}_{\beta}$ be a diagonal matrix whose diagonal elements are formed by taking the denominators of all the diagonal elements of $\mathbf{\Lambda}$.\footnote{{In the algorithm of calculating an lcrm of two non-singular integer matrices, all fractions used to generate these two diagonal matrices $\Lambda_{\alpha}$ and $\Lambda_{\beta}$ are in irreducible forms.}}
Then, we can get $\mathbf{\Lambda}_{\alpha}=\mathrm{diag}(1,q_{i}^{2})$ and $\mathbf{\Lambda}_{\beta}=\mathrm{diag}(q_{i}^{2},1)$. Therefore,
an lcrm of $\mathbf{M}_{i,1}$ and $\mathbf{M}_{i,2}$ is 
\[
\mathbf{M}_{i,1}\mathbf{U}^{-1}\mathbf{\Lambda}_{\alpha} = 
\mathbf{M}_{i,2}\mathbf{V}\mathbf{\Lambda}_{\beta}=
\begin{pmatrix}
    -q_{i}^{4}-q_{i}^{2} & q_{i}^{3}\\
    -q_{i}^{3} & q_{i}^{2}\\
\end{pmatrix},
\]
and $\det(\text{lcrm}(\mathbf{M}_{i,1},\mathbf{M}_{i,2}))=-q_{i}^{4}$.

Let $\mathbf{R}$ be any lcrm of matrices $\mathbf{M}_{i,1},\mathbf{M}_{i,2}$ for all $1 \leq i \leq L$. From 
(\ref{lcrm1}), $\mathbf{R}$ is also an lcrm of matrices $\text{lcrm}(\mathbf{M}_{i,1},\mathbf{M}_{i,2})$ for all $1 \leq i \leq L$. For the sake of convenience, let $\mathbf{M}_{i}=\text{lcrm}(\mathbf{M}_{i,1},\mathbf{M}_{i,2})$. There exists a nonsingular integer matrix $\mathbf{P}_{i}$ such that $\mathbf{R}=\mathbf{M}_{i}\mathbf{P}_{i}$ for each $1 \leq i \leq L$. Then, we can get
$$
|\det(\mathbf{R})|=|\det(\mathbf{M}_{i})||\det(\mathbf{P}_{i})|=q_{i}^{4}|\det(\mathbf{P}_{i})|,
$$
which means that $|\det(\mathbf{R})|$ has a divisor $q_{i}^{4}$, for every $1 \leq i \leq L$. Since $q_1,q_2,\cdots,q_L$ are pairwise co-prime integers,
we have that $|\det(\mathbf{R})|$ has a divisor $(q_1q_2\cdots q_L)^{4}$. Therefore, $|\det(\mathbf{R})| \geq (q_1q_2\cdots q_L)^{4}$.

We now show the following matrix
\begin{equation}\label{R2}
   \mathbf{R}_2=
\begin{pmatrix}
    (q_1q_2 \cdots q_L)^{2} & 0\\
    0 & (q_1q_2 \cdots q_L)^{2}\\
\end{pmatrix} 
\end{equation}
is an lcrm of matrices $\mathbf{M}_{i,1},\mathbf{M}_{i,2}$ for all $1 \leq i \leq L$. For each matrix $\mathbf{M}_{i,j}$, $1 \leq i \leq L$ and $1\leq j \leq 2$, matrix  
$\mathbf{M}_{i,j}^{-1}\mathbf{R}_2$
is a nonsingular integer matrix since
$$
\mathbf{M}_{i,j}^{-1}=\text{adj}{(\mathbf{M}_{i,j})}/\det(\mathbf{M}_{i,j})
\mbox{ and }
\det(\mathbf{M}_{i,j})=q_i^2.
$$
Therefore, it is a crm of matrices $\mathbf{M}_{i,1},\mathbf{M}_{i,2}$ for all $1 \leq i \leq L$. As $\det(\mathbf{R}_2)=(q_1q_2 \cdots q_L)^{4}$
and for any lcrm $\mathbf{R}$, its determinant absolute
value $\det(\mathbf{R})\geq (q_1q_2\cdots q_L)^4$ as proved above,
$\mathbf{R}_2$ has to be an lcrm of  
matrices $\mathbf{M}_{i,1},\mathbf{M}_{i,2}$ for all $1 \leq i \leq L$.
This proves the following theorem.

\begin{thm}\label{thm:4}
  The matrix $\mathbf{R}_2$ in (\ref{R2}) is an lcrm of the
  $2\times 2$ pairwise co-prime matrices in (\ref{D=2}) 
  constructed by Algorithm \ref{alg:1}. 
\end{thm}

We next study the case when the dimension $D$ is more than $2$, i.e., $D>2$.
For all $3\leq D\leq 75$, when  the feasible permutation set is the set of all cyclic permutations (\ref{cyc}), we have utilized Mathematica to perform symbolic calculations following the algorithm in \cite{smith2,MD2} and determined that for each group of matrices $\{\mathbf{M}_{i,j},1 \leq j \leq D\}$,
$$|\det(\text{lcrm}(\mathbf{M}_{i,j},1 \leq j \leq D))|=q_{i}^{D^2}.$$ 
Similar to the case of $D=2$, the determinant absolute value of an lcrm of all the matrices $\mathbf{M}_{i,j}$, $1\leq i\leq L$ and $1\leq j\leq D$,  in the constructed family
is greater than or equal to $(q_1q_2 \cdots q_L)^{D^2}$. 

Let  $\mathbf{R}_D$ be the
following $D\times D$ diagonal matrix:
\begin{equation}\label{RD}
   \mathbf{R}_D = (q_1q_2 \cdots q_L)^{D} \ \mathbf{I}. 
\end{equation} 
For each matrix $\mathbf{M}_{i,j}$, $1 \leq i \leq L$ and $1\leq j \leq L$,
matrix 
$\mathbf{M}_{i,j}^{-1}\mathbf{R}_D$
is a nonsingular integer matrix
since 
\[
  \mathbf{M}_{i,j}^{-1}=\text{adj}{(\mathbf{M}_{i,j})}/\det(\mathbf{M}_{i,j})
\mbox{ and }
\det(\mathbf{M}_{i,j})=q_i^D.  
\]
Therefore, it is a crm of all the matrices $\mathbf{M}_{i,j}$, $1 \leq i \leq L$ and $1\leq j \leq L$.
Since
\begin{equation}\label{detRD}
   \det(\mathbf{R}_D)=(q_1q_2 \cdots q_L)^{D^2}, 
\end{equation}
$\mathbf{R}_D$ has
to be an lcrm of all the matrices $\mathbf{M}_{i,j}$, 
$1 \leq i \leq L$ and $1\leq j \leq L$, 
in the constructed family.
This proves the following theorem.

\begin{thm}\label{thm5}
    For $3\leq D\leq 75$, the matrix $\mathbf{R}_D$ in (\ref{RD}) is an lcrm of the $D\times D$ pairwise co-prime matrices  
  constructed by Algorithm \ref{alg:1}, where the feasible permutation set is chosen to be the set of all cyclic permutations (\ref{cyc}).  
\end{thm}

For all $D>75$, due to our limited computational power,
although we are not able to confirm the above result, 
we have the following conjecture.
\begin{conjecture}  
  The matrix $\mathbf{R}_D$ in (\ref{RD}) is an lcrm of the $D\times D$ pairwise co-prime matrices  
  constructed by Algorithm \ref{alg:1} for $D>75$, where the feasible permutation set is chosen to be the set of all cyclic permutations (\ref{cyc}).
\end{conjecture}

If we can add any other integer
matrix $\mathbf{M}$ with $|\det(\mathbf{M})|=q_i^d$
for some integers $i$ and $d$, $1\leq i\leq L$ and $1\leq d\leq D$, 
which is not included in our constructed family,  to our constructed family,
$\mathbf{R}_D$ is still an lcrm of these $DL+1$ matrices, i.e.,
$\text{lcrm}(\mathbf{R}_D,\mathbf{M})=\mathbf{R}_D$,
since in this case, $\mathbf{R}_D$ is a right multiple of $\mathbf{M}$. 
This implies that for all $D\times D$ integer
matrices with their determinant absolute values 
$q_i^{d}$ for  $1 \leq i \leq L, 1\leq d\leq D$,
no matter they are pairwise co-prime or not,
$\mathbf{R}_D$ is an lcrm of all these $D\times D$ integer matrices.

In the meantime, from the above discussions, it is not hard to see that
if any member in our constructed family is removed  from the
family, the lcrm of the newly formed family of integer matrices
has strictly less determinant absolute value than $|\det(\mathbf{R}_D)|$,
since in this case  $q_i^D$ for some $i$, $1\leq i\leq L$, will not be included
in $|\det(\mathbf{R}_D)|$. Thus, we have proved the following corollary.

\begin{cor}\label{cor4}
  The constructed family $\mathbf{M}_{i,j}$, $1\leq i\leq L, 1\leq j\leq D$,
  is the smallest family of integer matrices
  with determinant absolute values $q_i^d$ for
  $1\leq i\leq L$ and $1\leq d\leq D$ such that
  $\mathbf{R}_D$ in (\ref{RD}) is their lcrm.
  \end{cor}

As mentioned in 7) and 8) of Section \ref{s2}, when $\mathbf{M}_{i,j}$ are used as sampling matrices, $|\det(\mathbf{M}_{i,j})|$ are  their sampling rates
and $|\det(\mathbf{R_{D}})|$ is their dynamic range. Thus,
under the same sampling rates, our constructed
family is the smallest set of
sampling matrices to achieve the maximal dynamic range. 

From the above results, we also
see that the number of integer vectors that can be
uniquely determined from their  integer vector remainders modulo
the constructed integer matrix moduli $\mathbf{M}_{i,j}$,
$1\leq i\leq L, 1\leq j\leq D$,
using MD-CRT, i.e., the dynamic range, 
is $(q_1q_2\cdots q_L)^{D^2}$. Notice that since these matrix moduli
$\mathbf{M}_{i,j}$,
$1\leq i\leq L, 1\leq j\leq D$, do not commute, 
 they cannot be diagonalized
simultaneously in any sense. Thus, the corresponding MD-CRT cannot
be equivalently converted to multiple conventional CRTs for integers.
In other words, they are non-separable.
On the other hand, to have the same dynamic range as
that using multiple individual CRT for integers, i.e.,
the separable case, it is obvious
to construct diagonal integer matrix moduli
as, for $1\leq i\leq L$, 
\begin{equation}\label{sep1}
  \mathbf{D}_{i,1}= \mbox{diag}(q_i^D,1,\cdots,1),\cdots,
  \mathbf{D}_{i,D}= \mbox{diag}(1,\cdots,1,q_i^D). 
  \end{equation}
These diagonal integer matrices are clearly pairwise co-prime
and their product is their lcrm,  and thus  their
dynamic range is also $(q_1q_2\cdots q_L)^{D^2}$, the same as
that of ${\bf M}_{i,j}$, $1\leq i\leq L, 1\leq j\leq D$.
In the meantime, $\det(\mathbf{D}_{i,j})=q_i^D=\det(\mathbf{M}_{i,j})$,
i.e., the sampling rates of
$\mathbf{D}_{i,j}$ and $\mathbf{M}_{i,j}$ are the same as well 
for $1\leq i\leq L, 1\leq j\leq D$. Note that only 
for the convenience in comparison,  these sampling rates
are counted in the sense of
 overall $D$ dimensional volume-wise sampling rates for the multi-dimensional
sampling matrices applied to the $D$ dimensional
real vector $(t_1,t_2,\cdots,t_D)^{\top}$.
Below we analyze the sampling rates for each dimension $t_j$,
i.e., how fast a sampling of each continuous real variable $t_j$ is.

We first see that the ratios of the peak values and the
average values of the components in integer matrices $\mathbf{M}_{i,j}$ and
integer matrices  $\mathbf{D}_{i,j}$ are, respectively, 
\begin{equation}\label{ratio1}
  \gamma_{\mathbf{M}_{i,j}} =\frac{D^2 q_i}{Dq_i+D-1}
  \approx D,  \mbox{ when }q_i \mbox{ are large}, 
\end{equation}
and
\begin{equation}\label{ratio2}
  \gamma_{\mathbf{D}_{i,j}} =\frac{D^2q_i^D}{q_i^D+D-1}
\approx D^2, \mbox{ when }q_i \mbox{ are large}. 
\end{equation}
Also,  the ratios of the peak values and the smallest non-zero absolute values 
 of the components in integer matrices $\mathbf{M}_{i,j}$ and
 integer matrices  $\mathbf{D}_{i,j}$ are, respectively,
 $q_i$ and $q_i^D$. 
 Clearly the above two ratios of  $\mathbf{M}_{i,j}$
 are smaller than those of $\mathbf{D}_{i,j}$.

 We next consider the sampling rates on each dimension.
 For the above diagonal sampling matrices
 ${\bf D}_{i,j}$, for each $i$ the diagonal element $q_i^D$ means that the sampling
 rate in dimension $j$ is $q_i^D$  that could be too high in
 practice when  $q_i$ and (or) $D$  are (is) large.

 For the  newly proposed sampling matrices $\mathbf{M}_{i,j}$ constructed
 in Algorithm \ref{alg:1}, where all the elements in $\mathbf{M}_{i,j}$ are
 non-negative, from 1) in Section \ref{s2} one can see that
 FPD   $\mathcal{N}(\mathbf{M}_{i,j})$
 (or the component-wise inverses of its elements) corresponds to 
 the unit spatial volume of $\mathbb{R}^{D}$ in the multi-dimensional sampling.
 We next show that for any dimension $k$, any line $\mathcal{L}_k$  included in
 set $\mathcal{N}(\mathbf{M}_{i,j})$,  that is parallel to the $k$-th
 dimensional coordinate axis of variable  $t_k$,
 has the largest value not above $q_i+1$ and
the smallest value not below $0$. 
 Since all elements in $\mathbf{M}_{i,j}$ are not negative,
 the smallest value on line $\mathcal{L}_k$ is not below $0$. Next, we show that the largest value on line $\mathcal{L}_k$ is not above  $q_i+1$.
 This means that the sampling rate in dimension
 $k$ for continuous variable $t_k$ is not larger than $q_i+1$ that is much
 smaller than the largest sampling rate $q_i^D$ for the above diagonal
 sampling matrix $\mathbf{D}_{i,j}$.

For line  $\mathcal{L}_k$ included in  $\mathcal{N}(\mathbf{M}_{i,j})$ that is parallel to the $k$-th dimensional coordinate axis, let the two ending points of this line
be  $\mathbf{x}_{l}=[x_{1},\cdots,x_{k,l},\cdots,x_D]^{\top}$ for $l=1,2$.
To show the largest value on the line $\mathcal{L}_k$ is smaller than $q_i+1$,
we only need to show that $x_{k,1}$ and $x_{k,2}$ are smaller than $q_i+1$.
From the definition of $\mathcal{N}(\mathbf{M}_{i,j})$
in 1) in Section \ref{s2}, there must be two vectors $\mathbf{a}=[a_1,a_2,\cdots,a_D]^{\top}$ and $\mathbf{b}=[b_1,b_2,\cdots,b_D]^{\top}$
in $[0,1)^D$ 
such that
\begin{equation}\label{samp1}
\mathbf{M}_{i,j}\mathbf{a}= \mathbf{x}_{1} \quad \text{and} \quad \mathbf{M}_{i,j}\mathbf{b}= \mathbf{x}_{2}.   
\end{equation}

From the construction of $\mathbf{M}_{i,j}$, the $k$-th row of $\mathbf{M}_{i,j}$ may have no $1$ or have a single $1$.

If the $k$-th row of $\mathbf{M}_{i,j}$ has no $1$, we can get $q_i a_k=x_{k,1}$ and $q_i b_k=x_{k,2}$ from (\ref{samp1}). Since $0\leq a_k,b_k < 1 $, we can obtain that $x_{k,1}$ and $x_{k,2}$ are smaller than  $q_i$.

If the $k$-th row of $\mathbf{M}_{i,j}$ has a single $1$, without loss of generality, let $\sigma_j=(m_1,\cdots,m_{D-1},m_{D})$. There must be an $i$,
$1 \leq i \leq D-1$, such that $k=m_i$. 
Then, from (\ref{samp1}) we can get  
\begin{equation}
    a_{m_{i+1}}+q_i a_{m_{i}} = x_{k,1}, \quad \text{and} \quad
    b_{m_{i+1}}+q_i b_{m_{i}} = x_{k,2}.
\end{equation}
Since $0\leq a_{m_{i+1}},a_{m_{i}},b_{m_{i+1}},b_{m_{i}} < 1 $, we conclude
that $x_{k,1}$ and $x_{k,2}$ are smaller than $q_i+1$. This completes the proof. 

 In summary, the above analysis
 tells us that for any $i$, $1\leq i\leq L$,
 the sampling rate in each dimension of the sampling
 matrix $\mathbf{M}_{i,j}$ is no larger than $q_i+1$  for any $j$,
 $1\leq j\leq D$,  while
 the sampling rate in one dimension of the sampling matrix
 $\mathbf{D}_{i,j}$ is  $q_i^D$.
 This shows an advantage of non-separable
  sampling over separable sampling for a multi-dimensional signal.
 Next, for the illustration convenience,
 we only take a two dimensional  example to show the sampling rate analysis result on each
 dimension.

\begin{example}\label{ex:1}
Let $q_i=3$, the diagonal $2\times 2$ sampling matrices are  
\[
\mathbf{D}_{i,1} = \begin{pmatrix} 9 & 0 \\ 0 & 1 \end{pmatrix} \quad \text{and} \quad \mathbf{D}_{i,2} = \begin{pmatrix} 1 & 0 \\ 0 & 9 \end{pmatrix},
\]
and the sampling matrices $\mathbf{M}_{i,j}$  are 
\[
\mathbf{M}_{i,1} = \begin{pmatrix} 3 & 0 \\ 1 & 3 \end{pmatrix} \quad \text{and} \quad \mathbf{M}_{i,2} = \begin{pmatrix} 3 & 1 \\ 0 & 3 \end{pmatrix}.
\]
The FPDs of these four matrices,
i.e., $\mathcal{N}(\mathbf{M}_{i,j})$ and
$\mathcal{N}(\mathbf{D}_{i,j})$,  are shown in Fig \ref{fig:ex}.
From the figure, we can easily
see that the maximal sampling rate  of $\mathbf{M}_{i,j}$
on each dimension 
is $4=q_i+1$, while the maximal sampling rate of $\mathbf{D}_{i,j}$
on each dimension is $9=q_i^2$.

\begin{figure}[htbp]
    \centering
    \includegraphics[width=\columnwidth]{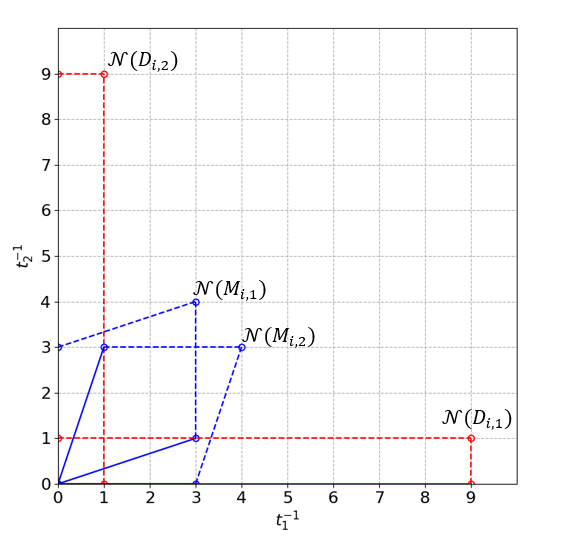}  
    \caption{{FPDs of the four matrices in Example \ref{ex:1}}}
    \label{fig:ex}  
\end{figure}
\end{example}

\section{Conclusion}\label{s6}

In this paper, we have presented a new construction of
pairwise co-prime integer matrices of any dimension and large size. They
are non-commutative and have low ratios of peak absolute values
over mean absolute values (or the smallest non-zero absolute values)
of their components. We have
also determined their least common right multiple (lcrm) with a closed and simple form.
These integer matrices have applications in MD-CRT
to determine integer vectors from their 
integer vector remainders, which may occur
in undersamplings of multi-dimensional
harmonic signals. Although the dynamic range of these non-diagonal 
integer matrices using MD-CRT can be achieved by diagonal integer
matrices using separable CRT for each dimension, 
their sampling rates  in each dimension
are much smaller than the conventional ones.
In other words, non-separable sampling has a  true
advantage over separable sampling for a multi-dimensional
signal.
This means that the newly constructed  pairwise co-prime integer matrices
may also have applications in multi-dimensional sparse
sensing and multi-dimensional
multirate systems. Since any new construction of
families of pairwise co-prime objects, such as co-prime
integers, co-prime algebraic numbers, and 
co-prime integer matrices, is fundamental, 
we believe that the new families of pairwise co-prime
integer matrices presented in
this paper may have other applications as well.

\end{document}